\documentclass[12pt]{article}
\usepackage{amsmath}
\usepackage{amsthm}
\usepackage{enumerate}
\usepackage{amssymb}
\usepackage{latexsym}
\newtheorem{theorem}{Theorem}[section]

\newtheorem{lemma}[theorem]{Lemma}
\date{}

\def\spacingset#1{\renewcommand{\baselinestretch}%
{#1}\small\normalsize} \spacingset{1}

  \title{\bf Estimation of Expected Shortfall under Various Experimental Conditions}
  \author{Jana Jure\v{c}kov\'a\thanks{The research was supported by the Grant  21-19311S of the Czech Science Foundation. The research of
	J. Jure\v{c}kov\'a was also supported by the Grant 22-036036S}\hspace{.2cm}\\
		The Czech Academy of Sciences, \\Institute of Information Theory
and Automation\\ and Charles University, Faculty of Mathematics and Physics, Prague\\
    and \\
    Jan Kalina \\
		The Czech Academy of Sciences, Institute of Computer Science\\
		and \\
		Jan Ve\v{c}e\v{r} \\
		Charles University, Faculty of Mathematics and Physics, Prague}	
\begin{document}			
  \maketitle
% }\fi
 
\date{}
%\if1\blind
%{\bigskip
 %\bigskip
 %\bigskip
% \begin{center}
%{\LARGE\bf {Estimation of Expected Shortfall under Various Experimental Conditions}}
%\end{center}
 %\medskip
%0}% \fi

\bigskip
\begin{abstract}
%The text of your abstract.  200 or fewer words.
Our primary aim is to find an estimate of the  expected shortfall in various situations: 
(1) Nonparametric situation, when the probability distribution of the incurred loss is unknown, only satisfying some general conditions. Then, following  \cite{Bassett}, the expected shortfall can be expressed through a minimization of a~well known quantile criterion and its numerical estimate is based on the empirical quantile function of the loss. 
(2) The distribution function of the loss is known, but the loss can be contaminated by an additive measurement error: Estimating the expected shortfall in such a case exploits the concept of pseudo-capacities elaborated in \cite{HuberStrassen} and \cite{Buja1986} and its numerical value is based on the empirical quantile function of the suitable capacity.	
(3) The loss distribution can be contaminated by the heavy right tail with Pareto index $\gamma>1$. The problem of interest is in this case to evaluate the effect of the Pareto index on the resulting  expected shortfall.
\end{abstract}

\noindent%
{\it Keywords:} expected shortfall; Choquet capacity; distortion function; additive measurement error; heavy tail distribution %3 to 6 keywords, that do not appear in the title
\vfill

%Jureckova, Jana. (2006). Quantile Regression. Roger Koenker. Journal of the American Statistical Association. 101.\newpage
\spacingset{1.8} % DON'T change the spacing!

\section{Introduction} %%%%%%%%%%%%%%%%%%%%%%%%%%%%%%%%%%%%%%%%%%%%%%%%%%%%%%%%%%%%%%%%%%%%%%%%%%%%%%%%%%%%%%%%%%%%%%%%%%%%%%%%%%%% Section 1
\label{sec:intro}

(Financial) decision making often struggles with various uncertainties; the accepted  conclusions can suffer from a risk. We would like to  predict our risk of a loss before making a~decision. There is a rich literature on this subject;  several authors developed various measures of risk, with a profound mathematical background. In the financial sector one usually considers the portfolio risk, which can be taken  as a negative utility. We prefer to follow \cite{Wirch}, who argued that the risk measure should reflect the solvency in case of a unfavorable experience, thus  the loss distribution should be censored at zero. \cite{Artzner} introduced the class of coherent risk measures possessing four basic properties: monotonicity, homogeneity, sub-additivity, and translation invariance. The additivity or even the $\sigma$-additivity is released to sub-additivity,  which is more flexible in decisions under uncertainty. The flexibility is enabled by involving the concept of a non-additive probability founded by \cite{Choquet}. 

The risk measures are typically functionals of the quantile function of the loss  and  can be estimated by its empirical quantile function. When there is a whole family of possible probability distributions, which can be dominated by a suitable Choquet capacity, then  the risk measures can be based on the  capacity or on the least favorable distribution of the family (cf.~\cite{HuberStrassen}). 
Under uncertainty or in the presence of measurement errors, we try to cover the system of possible distributions by a suitable capacity, on which we can build the inference.  The capacity induces a coherent risk measure; it is generally a non-additive set function of events, which leads to a wider flexibility.  

We shall illustrate this situation on estimating a specific coherent risk measure, namely the expected shortfall.
This will be elaborated in the situation with an unknown probability distribution of the incurred loss, or in the setup with unobservable additive measurement error. In the latter case, we cover the model with a suitable capacity, which in turn is a probability measure, and then calculate the expected shortfall corresponding to this capacity. Moreover, we shall illustrate how the risk measure depends on the tail index of the loss distribution. 

After a discussion of the concept of the Choquet expected utility of a quantile functional (Section 2), we study the nonparametric estimation of expected shortfall (Section 3). 
In Section 4, expected shortfall is estimated in a model with additive measurement errors contaminating the known distribution of the loss; the estimation exploits covering the family of possible distributions with a suitable capacity. 
Section 5 illustrates the effect of heavy tails of the loss distribution on the expected shortfall. Numerical illustrations accompany the theoretical results.

\section{Choquet expected utility of a quantile functional} %%%%%%%%%%%%%%%%%%%%%%%%%%%%%%%%%%%%%%%%%%%%%%%%%%%%%%%%%%%%%%%%%%%%%%%%%%%%% Section 2
\setcounter{equation}{0}

The set function $w$ on the measurable space $\left(\Omega, \mathcal B\right)$ is defined as a pseudo-capacity (\cite{Buja1986}), if it satisfies
\begin{equation}\label{capacity} 
\begin{array}{ll}
%\begin{description}
	\mathbf {(a)} &w(\emptyset)=0, \; w(\Omega)=1\\
	\mathbf{(b)} &w(A)\leq w(B) \quad \forall A\subset B\\
	\mathbf{(c)}  &w(A_n)\uparrow w(A) \quad \forall A_n\uparrow A\\
	\mathbf{(d)}  &w(A_n)\downarrow w(A) \quad \forall A_n\downarrow A\neq \emptyset\\
	\mathbf{(e)} &w(A\cup B)+w(A\cap B)\leq w(A)+w(B).
\end{array}	
%\end{description}
\end{equation}

Consider the  random variable $X: \Omega\mapsto\mathbb R$ with a non-atomic probability distribution~$\mathcal P$, distribution function $F$, hazard function $\bar{F}=1-F$, and density $f.$  We interpret $X$ as a~loss and consider only values $X\geq 0$ with a positive probability. The Choquet expectation of the random variable $X$ with respect to the capacity $w$ is defined as
\begin{equation}\label{choquet}
\mathbf E_w X=\int_0^{\infty}w\Big(\{\omega: X(\omega)\geq x\}\Big)dx+\int_{-\infty}^0\Big[w\Big(\{\omega:X(\omega)\geq x\}\Big)-1\Big]dx.
\end{equation}
If there exists a non-decreasing  function $\Phi:[0,1]\mapsto[0,1]$ (the distortion function) such that $\Phi(0)=0$,  $\Phi(1)=1$ and $w\left(\omega: X(\omega)>x\right)=\Phi(1-F(x))$, then the Choquet expectation (\ref{choquet}) can be rewriten as 
\begin{equation}\label{choquet1}
\mathbf E_w X=\int_0^{\infty}\Phi(1-F(x))dx+\int_{-\infty}^0\Big[\Phi(1-F(x))-1\Big]dx. 
\end{equation}

The formula (\ref{choquet1}) is called the \textit{Choquet expected loss} and $w$ is then denoted as a \textit{distorted measure}.    
The functionals of type  (\ref{choquet1}) represent a general class of coherent  measures in the sense of \cite{Artzner}. 
We concentrate on the popular group of distortion functions of the form
\begin{equation}\label{choquet3}
\Phi_{\alpha}(t)=\left\{ \begin{array}{lll}
\frac{t}{1-\alpha} &  \ldots  & 0<t\leq 1-\alpha \\[3mm]
1                            & \ldots  &  1-\alpha<t\leq 1.\\
\end{array}\right .                    
     %\min\left\{t/\alpha, \; 1\right\}, \quad 0\leq\alpha\leq 1.
\end{equation}
Then,
\begin{equation}\label{shortfall}
\mathbf E_{\Phi_{\alpha}}(X)={\sf CVaR}_{\alpha}(X)=(1-\alpha)^{-1}\int_{\alpha}^1 F^{-1}(t)dt.
\end{equation}
The measure (\ref{shortfall}) is commonly denoted as {\it expected shortfall} (as e.g.~in \cite{Acerbi}).
We use the abbreviation {\sf CVaR} corresponding to the alternative name {\it conditional value at risk}, used e.g.~in \cite{Urysaev}. Other alternative names for~(\ref{shortfall}) include 
{\it tail conditional expectation} \cite{Artzner} or {\it $\alpha$-risk of the random prospect $X$} \cite{Bassett}.

Our main aim is to estimate the expected shortfall ${\sf CVaR}_{\alpha}$ in {various realistic situations}: 
\begin{description}
\item[(1)] When the distribution function $F$ is unknown, only under some general conditions (nonparametric situation). By \cite{Bassett}, the expected shortfall can be then expressed through a minimization of a well known quantile criterion. If we  have independent  
  observations  $X_1,\ldots,X_n$ of $X$ at disposal, then ${\sf CVaR}_{\alpha}$ can be numerically estimated with the aid of the empirical quantile function of~$X$. 
\item[(2)]  The situation when the distribution function $F$ of $X$ is known, but $X$ is contaminated by an additive measurement error;  we only have observations of 
  \begin{equation}\label{3}
	Z_{\delta}=X+\sqrt{\delta}V.
  \end{equation}		
  The independent values $X_1,\ldots,X_n$ are not directly observable and the only available observations are $Z_{i,\delta}=X_i+\sqrt{\delta}V_i, \;  i=1,\ldots,n.$  Here, $V_1,\dots,V_n$ are unobservable iid random variables independent of $X_i, \;  i=1,\ldots,n,$ and 
	$\delta>0$ is an unknown parameter. We shall  assume that the  distribution function $G$ of  $V_i$ is symmetric, otherwise unknown,  and that $\mathbf EV_i=0$ and $\mathbf E V_i^2=1.$ Because only $Z_{\delta}$ is observable, we can empirically estimate only  
  ${\sf CVaR}_{\alpha}(Z_{\delta})$  for contaminated observations, and only asymptotically for $\delta\downarrow 0.$ The model $Z_{\delta}=X+\sqrt{\delta}V$ has been studied by \cite{Guo} and references cited therein. Guo investigated the behavior of 
	various divergencies between two models including the Kullback-Leibler.
\item[(3)]  The supposed loss distribution $F_0$ can be contaminated by a heavy tail with the Pareto index $\gamma>1.$ Then, the question of interest is to evaluate the effect of $\gamma$ on the expected shortfall.
\end{description}
While the problems are illustrated on the conditional measure at risk, similar situations may concern other risk measures. 
  
\section{Nonparametric estimation of expected shortfall} %%%%%%%%%%%%%%%%%%%%%%%%%%%%%%%%%%%%%%%%%%%%%%%%%%%%%%%%%%%%%%%%%%%%%%%%%%%%%%%%%%%%%% Section 3
\setcounter{equation}{0}

Several nonparametric estimators of the expected shortfall have appeared in the literature. 
Some of them were recalled in a recent work of \cite{Fang}, who introduced a~weighted single index quantile regression as a natural extension of the single index quantile regression.
Let us now proceed to proposing a novel estimator in this section.

Let $F$ and $f$  be the  distribution function and density of  the loss $X,$ generally unknown.
Denote the function
\begin{equation}\label{rho}
\rho_{\alpha}(x)=x\left(\alpha-I[x<0]\right), \quad  x\in\mathbb R.
\end{equation}
It has been shown in \cite{Bassett} [Theorem 2] that
\begin{equation}\label{B4}
{\sf CVaR}_{\alpha}(X)=(1-\alpha)^{-1}\int_{\alpha}^1 F^{-1}(t)dt=(1-\alpha)^{-1}\min_{\xi\in\mathbb R}\rho_{\alpha}(X-\xi)+{\mathbf E}X.
\end{equation}
It is well known that the solution of the minimization $$\min_{\xi\in\mathbb R}\rho_{\alpha}(X-\xi)$$ is the $\alpha$-quantile of $X.$
Hence, if $F$ is unknown, the estimate of ${\sf CVaR}_{\alpha}(X)$ can be obtained from  the empirical quantile function based on independent observations $X_1, X_2,\ldots, X_n$ of~$X$.
The consistent estimate of ${\sf CVaR}_{\alpha}(X)$ is a version of the $\alpha$-trimmed mean based on the order statistics $X_{n:1}\leq X_{n:2} \leq\ldots\leq X_{n:n}$.

\begin{theorem} 
Under the above conditions,  the estimate of the $\alpha$-expected shortfall of $X$ has the form 
\begin{equation}
\label{B4a}
\widehat{\sf CVaR}_{\alpha}(X)=\frac{1}{\left\lfloor n(1-\alpha)\right\rfloor} \sum _{i=\left\lfloor n(1-\alpha)\right\rfloor }^n X_{n:i}.
%-\frac{1}{\left\lceil n\alpha\right\rceil}\sum_{i=1}^{\left\lceil n\alpha\right\rceil} X_{n:i}.
\end{equation}
\label{th:1}
\end{theorem}

\begin{proof} Because of (\ref{B4}), we have
$$\widehat{\sf CVaR}_{\alpha}(X)= {\left\lfloor n(1-\alpha)\right\rfloor}^{-1} \min_{\xi\in\mathbb R}\sum_{i=1}^n \rho_{\alpha}(X_{i}-\xi)+\bar{X}_n,$$
%\left\lceil n\alpha\right\rceil^{-1}\min_{\xi\in\mathbb R}\sum_{i=1}^n \rho_{\alpha}(X_{i}-\xi)-\bar{X}_n$$
where $\bar{X}_{n}=n^{-1}\sum_{i=1}^nX_{i}$.
The minimum is attained for $\xi=X_{n:\left\lceil n\alpha\right\rceil}$, hence
\begin{eqnarray*}
&&\widehat{\sf CVaR}_{\alpha}(X)=\bar{X}_{n} 
+{\left\lfloor n(1-\alpha)\right\rfloor}^{-1} \sum_{i=1}^n(X_{i}-X_{n:\left\lceil n\alpha\right\rceil})%\left(-1+\alpha\right\rceil\right)
(-1+\alpha+I[X_i>X_{n:\left\lceil n\alpha\right\rceil}])\\
%\left\lceil n\alpha\right\rceil^{-1}\sum_{i=1}^n(X_{i}-X_{n:\left\lceil n\aa-I[X_i\geq X_{n:\left\lceil n\alpha\right\rceil}]\right)\\
&&=X_{n:\left\lceil n\alpha\right\rceil}+{\left\lfloor n(1-\alpha)\right\rfloor}^{-1}\sum_{i={\left\lfloor n(1-\alpha)\right\rfloor}}^n (X_{n:i}-X_{n:\left\lceil n\alpha\right\rceil})\\
%%-\left\lceil n\alpha\right\rceil^{-1}\sum_{i=1}^{\left\lceil n\alpha \right\rceil}X_{n:i} +X_{n:\left\lceil n\alpha\right\rceil}
&&={\left\lfloor n(1-\alpha)\right\rfloor}^{-1}\sum_{i={\left\lfloor n(1-\alpha)\right\rfloor}}^n X_{n:i}.
\end{eqnarray*}
%\hfill$\Box$\\
\end{proof}

\section{Expected shortfall under measurement errors} %%%%%%%%%%%%%%%%%%%%%%%%%%%%%%%%%%%%%%%%%%%%%%%%%%%%%%%%%%%%%%%%%%%%%%%%%%%%%%% Section 4
\setcounter{equation}{0}

Consider the expected shortfall in the situation with $X$ contaminated by an additive measurement error, when we can only observe $Z_{\delta}=X+\sqrt{\delta}V$ with unknown $V$ and unknown $\delta>0.$ 
Denote $F_{\delta}$ and $f_{\delta}$ the distribution function and density of $Z_{\delta};$ we shall occasionally use the notation $F=F_0$ and $f=f_0.$ 
Then, it remains impossible to use Theorem~\ref{th:1} in practice, because $\delta$, $X_i$ and $V_i$ for $i=1,\dots,n$ are unknown.
However, the contamination of $X$ by $\sqrt{\delta}V$ with unknown $\delta$ and $V$ leads not only to one, but to a whole family of probability distributions of $Z_{\delta}$ which we can try to cover by a~suitable capacity. 

Let $G$ and $g$ denote the distribution function and density of $V,$ respectively.
Assume that  $f_0$ has differentiable and integrable derivatives up to order 4. Notice that if densities of $X$ and $V$ are strongly unimodal, then the density of $Z_{\delta}$ is also strongly unimodal (see \cite{Ibragimov}).  We can approximate $F_{\delta}$ and 
$f_{\delta}$ for small $\delta>0$ in two ways: The less precise approximation %of $F_{\delta}$ up to $O(\delta)$ 
does not depend on the shape of $V,$ while a more precise approximation of $F_{\delta}$ depends on the kurtosis of $V,$ if  we know that ${\mathbf E} V^4<\infty.$   Based on that, we finally obtain an~approximation of ${\sf CVaR}.$ 

\begin{lemma}%\label{Lemma2}
Assume that $f$ is strongly unimodal and has differentiable and integrable derivatives up to  order 4. Moreover, assume that $V$ is symmetrically distributed, 
${\mathbf E}V=0,$ ${\mathbf E}V^2=1$, and ${\mathbf E}V^4<\infty.$ Then, as $\delta\downarrow 0,$    
\begin{equation}\label{G54}
F_{\delta}(z)=P\left(X+\sqrt{\delta}V\leq z\right)=F(z)+\frac{\delta}{2}f^{\prime}(z)+\frac{\delta^2}{4!}f^{(3)}(z)\mathbf E(V^4)+o(\delta^2) 
\end{equation}
and
\begin{equation}\label{G54a}
f_{\delta}(z)=f_0(x+\sqrt{\delta}V)=f_0(x)+\frac{\delta}{2}\frac{d^2}{dz^2}f_0(x)+\frac{\delta^2}{4!}\frac{d^4}{dz^4}f_0(x)\mathbf E(V^4)+o(\delta^2). 
\end{equation} 
\end{lemma}

\begin{proof} 
Under the assumptions on $F$ and $V$, we can write for $z\in\mathbb R$
\begin{eqnarray*}
&&P(X+\sqrt{\delta}V\leq z)=\int F(z-\sqrt{\delta}v)dG(v)\\
&&=\int \left[F(z)-\sqrt{\delta}v f(z)+\frac{1}{2}\delta v^2 f^{\prime}(z)-\frac{1}{3!}\delta^{3/2}v^3 f^{\prime\prime}(z)+
\frac{1}{4!}\delta^2 v^4 f^{(3)}(z) \right]dG(v)+o(\delta^2)\\
&&=F(z)+\frac{\delta}{2}f^{\prime}(z)+\frac{\delta^2}{24}f^{(3)}(z){\mathbf E}V^4+o(\delta^2).      
\end{eqnarray*}
Moreover, (\ref{G54a}) can be derived  with the aid of characteristic function (see \cite{Kalina}). 
\end{proof}

The modeled distribution $f$ can be asymmetric with steeper peaks and heavier tails but unimodal with finite moments. Because our true observations are contaminated as $Z_i=X_i+\sqrt{\delta}V_i, \; i=1,\ldots,n$, our predicted risk measure will be determined only by the 
$Z_i$ and $F_{\delta},$ even if our modeled $f_0$ is right. Hence, following (\ref{choquet1}), ${\sf CVaR}_{\alpha}$ with the distortion function (\ref{choquet3}) will take on the form
\begin{equation}\label{4}
{\sf CVaR}_{\delta,\alpha}=\frac{1}{1-\alpha}\int_{\alpha}^1 F_{\delta}^{-1}(u)du,
\end{equation}
where the notation ${\sf CVaR}_{\delta,\alpha}$ is now used to stress the dependence on $\delta$.
As in \cite{Kalina}, we consider the family $\mathcal H$ of distributions $\left\{F_{\delta,\kappa}(\cdot), \; \delta\leq\Delta, \; \kappa\leq K\right\}$. Then, the set function on the Borel  $\sigma$-field $\mathcal B$ 
\begin{equation}\label{33} 
v(B)=\left\{ \begin{array}{lll}
\sup\left\{ F(B): F\in \mathcal H\right \}&\rm{if}& B\neq \emptyset\\[5mm]
0 &\rm{if}& B=\emptyset\\
\end{array}
\right. 
\end{equation} 
is a pseudo-capacity in the sense of \cite{Buja1986} (see (\ref{capacity})).

Specifically, we assume that $X$ and $V$ are independent and that $\mathbf E V=0, \; \mathbf E V^2=1$, and $\mathbf E V^4<\infty.$ Moreover, we assume that $f_0$ and $g_0$ are symmetric, strongly unimodal and differentiable up to order 4, with  derivatives integrable and increasing distribution functions $F_0$ and $G_0,$ respectively. We state the range for the kurtosis of the measurement errors $V$ in the form
\begin{equation}\label{kurt}
1\leq \mathbf E V^4\leq K
\end{equation}
with a fixed $K, \; 0<K<\infty.$ %It means that the distribution of $V$ is restricted to have the  tails lighter than $t$-distribution with 4 degrees of freedom. We shall apply the 2-alternating capacity around specific subfamily of densities $f_0$, proposed in \cite{Kalina}. 
We shall concentrate on the family $\mathcal H^*$ of densities defined as
\begin{equation}\label{34}	 
\quad \mathcal H^*=\left\{ f_{\delta,\kappa}^*: \; f_{\delta,\kappa}^*(z)=f_0(z)+\frac{\delta}{2}f_0^{\prime\prime}(z)+\kappa\frac{\delta^2}{24}f_0^{(4)}(z) \; \Big| \: 0<\delta\leq\Delta, 1\leq \kappa\leq K\right\}
\end{equation}
with suitable fixed $\Delta, K>0.$ Then, under our assumptions, each $f_{\delta,\kappa}^*\in \mathcal H^*$ is a positive and symmetric density satisfying
\begin{equation}
\sup_{\delta\leq \Delta, \kappa\leq K}\sup_{z\in\mathbb R}\left| f_{\delta,\kappa}^*(z)-f_{0}(z)\right |\leq  \frac{CK}{12} \; \Delta^2+o(\Delta^2).
\end{equation}

Let $F_{\delta,\kappa}^*(B)$ be the probability distribution induced by density $f_{\delta,\kappa}^*\in \mathcal H^*$ for $B\in \mathcal B$, where $\mathcal B$ is the Borel $\sigma$-algebra. Then, the set function 
\begin{equation}\label{32} 
w(B)=\left\{ \begin{array}{lll}
\sup\left\{ F^*(B): F\in \mathcal H^*\right \}&\rm{if}& B\neq \emptyset\\[5mm]
0 &\rm{if}& B=\emptyset\\
\end{array}
\right. 
\end{equation}
is a pseudo-capacity in the sense of \cite{Buja1986}. %It is a probability measure with the distribution function
%\begin{equation}\label{e:v}
%\widehat{\mathbf F} (z)=w\left( (-\infty, z ]  \right) = \sup \{ F_{\delta,\kappa}^*(z); F_{\delta,\kappa}^* \in {\mathcal H}^* \}, \quad z %\in \R
%\end{equation} 
%and
Further, $w(-\infty,x]\geq F(x)$ for every $x\in\mathbb R$ and $F\in\mathcal H^*.$ Hence,
\begin{equation}
w^{-1}(u)\leq F^{-1}(u)\quad \forall ~0<u\leq 1, \; \forall F\in\mathcal H^*.
\end{equation}
%\begin{equation}\label{e:v1}
%F^{-1}(\alpha)\geq \widehat{\mathbf F}^{-1}(\alpha) \quad  \forall \alpha\in(0,1)
%\end{equation}
%for every $F\in{\mathcal H}^*.$ 
It means that 
\begin{equation}\label{35} 
{\sf CVaR}_{\alpha}(F)\geq {\sf CVaR}_{\alpha}(w) \quad \forall \alpha\in(0,1) \; \mbox{ fixed and } \; \forall F\in\mathcal H^*.
\end{equation}
In other words, the conditional $\alpha$-measure of risk cannot be smaller than ${\sf CVaR}_{\alpha}(w)$ for~any $F\in\mathcal H^*$, i.e.~smaller than ${\sf CVaR}_{\alpha}(F^*)$, where $F^*$ is the least favorable distribution in~$\mathcal H^*$. 

\subsection{Numerical illustration}%experiments: \rm{CVAR} under measurement errors} %%%%%%%%%%%%%%%%%%%%%%%%%%%%%%%%%%%%%%%%% 4.1
\label{sec:num1}

Let us have observations (measurements) $X_1, \ldots, X_n$ following the standard normal distribution $N(0, 1)$. These may be interpreted as the payoffs of a given portfolio. 
As above, we consider the measurement error model $Z_\delta = X + \sqrt{\delta}V$ assuming $\delta\in [0, \Delta]$ with a known~$\Delta$.
The distribution of $V$ is unknown, but we assume $\mathbf E V = 0,  \mathbf E V^2 = 1,$ and $\mathbf EV^4 \in [1, K]$ with a known $K.$ 
Let us use the notation $\psi$ for the density of $N(0, 1)$ distribution and $\Psi$~for the corresponding cumulative distribution function. %$\Psi^{-1}$ for the corresponding quantile function.

In order to resort to an approximation based on pseudo-capacities, let us consider the function
\begin{equation}\label{4a}
F_{\delta,\kappa}(z)=\Psi(z)+\frac{\delta}{2}\psi^{\prime}(z)+\kappa\frac{\delta^2}{24}\psi^{(3)}(z), \quad z\in{\mathbb R},
\end{equation}
expressed for a fixed $\delta\in[0, \Delta]$ and $\kappa\in[1, K].$ In this situation, we have
\begin{eqnarray}
\psi^{\prime}(z) &=& (-z) \psi(z), \; z\in\mathbb R,\\
\psi^{(3)}(z)&=&(-z^3+ 3z)\psi(z), \; z\in\mathbb R.\nonumber
\end{eqnarray}
In the model (\ref{3}) with an additive measurement error, using the expansion  (\ref{4a}) of $F$, we obtain an upper bound for the value of ${\sf CVaR}_\alpha$ in the form
\begin{equation}\label{6}
{\sf CVaR}_\alpha \leq \sup_{\delta\in[0,\Delta], \kappa\in[1,K]}\left[\frac{1}{1-\alpha}\int_\alpha^1 (F_{\delta,\kappa}^{\ast})^{-1}(t)dt\right].
\end{equation}

The values of the upper bound for ${\sf CVaR}_\alpha$ in (\ref{6}) for various choices of $\Delta$ and $K$ and for $\alpha= 0.04$ are given in Table 1.
The upper bounds increase with an increasing $\Delta$ and slightly (negligibly) increase with an increasing $K$.

\begin{table}[h]
\caption{Illustration of Section~\ref{sec:num1}. Values of the upper bound for ${\sf CVaR}_\alpha$ given by (\ref{6}) for various choices of $\Delta$ and $K$ and for $\alpha=0.04$.}
\label{tab:1}
\centering
\begin{tabular}{l|ccc}
& $K=1$ & $K=1.1$ & $K=1.2$\\\hline
$\Delta=0$    & 2.154 & 2.154 & 2.154\\
$\Delta=0.05$ & 2.206 & 2.207 & 2.207\\
$\Delta=0.10$ & 2.255 & 2.256 & 2.256\\
$\Delta=0.15$ & 2.302 & 2.303 & 2.303\\
$\Delta=0.20$ & 2.347 & 2.347 & 2.347\\
\end{tabular}\\
\end{table}

Specifically, if the $N(0,1)$ distribution is replaced with $N(0,\sigma^2)$, then (\ref{4a}) becomes
\begin{equation}
F_{\delta,\kappa}(z)=F(z)+\frac{\delta}{2}f^{\prime}(z)+\kappa\frac{\delta^2}{24}f^{(3)}(z), \; z\in\mathbb R,
\end{equation}
with $F, f$ being the distribution function and density of $N(0,\sigma^2)$, respectively, and 
\begin{equation}
f^{\prime}(z)=-\frac{z}{\sigma^2}f(z), \quad f^{(3)}(z)=\frac{3z\sigma^2-z^3}{\sigma^6}f(z).
\end{equation}

\section{Expected shortfall under heavy-tailed distribution} %%%%%%%%%%%%%%%%%%%%%%%%%%%%%%%%%%%%%%%%%%%%%%%%%%%%%%%%%%%%%%%%%%%%%%%%%%%%%%%%%%%%%%%%%%%%%%%% Section 5 
\setcounter{equation}{0}

\cite{HuberStrassen} mentioned the system of probability measures
\begin{equation}\label{C1}
\mathcal P=\{P=(1 - \varepsilon)P_0+\varepsilon Q \;| \; Q\in\mathcal M\},
\end{equation}
with $\mathcal M$ being a family of probability measures on $\mathcal B.$ This corresponds to a contaminated environment of probability distribution $P_0,$  or to the system of distribution functions
\begin{equation}\label{C2}
\mathcal F=\{F=(1 - \varepsilon)F_0+\varepsilon H \; | \; H\in \mathcal H\},
\end{equation}
where $\mathcal H$ is a family of distribution functions. 
\cite{HuberStrassen} considered the set function 
\begin{eqnarray}\label{huber}
v(A) &=& (1 - \varepsilon)P_0(A)+\varepsilon \; \mbox{ for } \; A\neq \emptyset\\
 v(\emptyset) &= & 0 \nonumber
\end{eqnarray}
as a suitable capacity for  model (\ref{C1}). 
However, notice that $v(A)\geq \varepsilon>0$ for $A\neq\emptyset,$ which is not convenient for construction of a coherent risk measure. 

\subsection{Contamination by the heavy tail of the Pareto distribution}%%%%%%%%%%%%%%%%%%%%%%%%%%%%%%% 5.1
\label{sec:pareto}

Instead of (\ref{C1}), we propose an environment of a specific measure $P_0$ [with continuous monotone distribution function $F_0$] contaminated by the heavy tail of the Pareto distribution function
\begin{equation}\label{Pareto}
G_{\gamma,\alpha}(x)=\left\{\begin{array}{lll}
0&\ldots&x\leq A_{\gamma,\alpha}\\[3mm] % (1-\alpha)^{1/\gamma}F_0^{-1}(\alpha)\\
1-(1-\alpha)\left(F_0^{-1}(\alpha)/x\right)^{\gamma}&\ldots&x>  A_{\gamma,\alpha},\\ % (1-\alpha)^{1/\gamma}F_0^{-1}(\alpha) .\\
\end{array}\right.
\end{equation}
starting at 
\begin{equation}\label{Pareto1}
A_{\gamma,\alpha}=(1-\alpha)^{1/\gamma}F_0^{-1}(\alpha)>0,  \quad \gamma>1,  \;  F_0(0)\leq\alpha<1.  %\tau_{\alpha}, \qquad \tau_{\alpha}= 
\end{equation}
Notice that 
\begin{equation}
G_{\gamma,\alpha}(F_0^{-1}(\alpha))=\alpha.
\end{equation}
Consider the family $\mathcal F$ of distribution functions  
\begin{equation}\label{contam}
F_{\gamma,\alpha}(x)=\left\{ \begin{array}{lll}
%F_0(x)+\varepsilon G_{\gamma,\alpha}(x)=
F_0(x)&\ldots&x\leq F_0^{-1}(\alpha)\\
G_{\gamma,\alpha}(x)&\ldots&x > F_0^{-1}(\alpha)\\
%F_0(x)+\varepsilon\left(1-\tau_{\alpha}^{\gamma}x^{-\gamma}\right)&\ldots&x> \tau_{\alpha}.\\%&\ldots&x\leq\tau_{\alpha}.\\
\end{array}\right .   \qquad
\gamma>1, \; F_0(0)\leq\alpha<1.
\end{equation}

Let us find the expected shortfall of $F_{\gamma,\alpha}$, which is influenced by the Pareto tail, i.e.~by the value of the Pareto index $\gamma$.
It is of interest to compare it with ${\sf CVaR}_{F_0,\alpha}$, where the latter is the expected shortfall evaluated for the non-contaminated $F_0$.

\begin{theorem}
\label{th:pareto} 
The conditional measure of risk corresponding to the loss $X,$ %Choquet expectation of $X$ 
distributed according to $F_{\gamma,\alpha},$ equals to %respect to the capacity $w_{\alpha}$, induced by distribution function $\widetilde{F}_{\gamma,\alpha}$, 
\begin{eqnarray} \label{w1}
%{\E}_{w_{\gamma,\alpha}}X&=& \E X+\varepsilon F_0^{-1}(\alpha)\frac{1-\alpha}{\gamma-1}
{\sf CVaR}_{\gamma,\alpha}&=& (1-\alpha)^{-1}\int_{\alpha}^1 F_{\gamma,\alpha}^{-1}(t)dt=
(1-\alpha)^{-1}\int_{\alpha}^1 F_0^{-1}(t)dt+(1-\gamma)^{-1}F_0^{-1}(\alpha)\nonumber\\
&=&{\sf CVaR}_{F_0, \alpha}+(1-\gamma)^{-1}F_0^{-1}(\alpha),
\end{eqnarray}
hence it is increasing with $\gamma.$
\end{theorem}

\begin{proof} Denote $\tau_{\alpha}=F_0^{-1}(\alpha)$ and assume that $\tau_{\alpha}\geq 0$. Then 
$$1-F_{\gamma,\alpha}(x)=\left\{\begin{array}{lll}
1-F_0(x) &\ldots & x\leq\tau_{\alpha}\\
(1-\alpha)\left(\tau_{\alpha}/x\right)^{\gamma} & \ldots & x>\tau_{\alpha}\\
\end{array}\right.$$
The conditional measure of $X_{\gamma,\alpha}$ is then
\begin{eqnarray*}
&&\mathbf E_{F_{\gamma,\alpha}}(X)=\int_{-\infty}^0 (-F_0(x))dx+\int_{0}^{\tau_{\alpha}}(1-F_0(x))dx
+(1-\alpha)\tau_{\alpha}^{\gamma}\int_{\tau_{\alpha}}^{\infty}x^{-\gamma}dx\\
&&=\int_0^{\tau_{\alpha}}xdF_0(x)-(1-\gamma)^{-1}(1-\alpha)\tau_{\alpha}\\
&&=\int_0^{\alpha}F^{-1}(t)dt-(1-\gamma)^{-1}(1-\alpha)F_0^{-1}(\alpha),
\end{eqnarray*}
which already implies (\ref{w1}). 
\end{proof} 

%\subsection{Numerical illustration of Section~\ref{sec:pareto}.} %%%%%%%%%%%%%%
\noindent {\bf Numerical illustration.} 
Values of ${\sf CVaR}_{\gamma,\alpha}$ of Theorem~\ref{th:pareto} will be now illustrated on a~numerical example.
Let us assume $F_0$ to be the distribution function of $\chi^2_1$ distribution.
The values of ${\sf CVaR}_{\gamma,\alpha}$ for various values of $\gamma$ and $\alpha$, which are shown in Table~\ref{tab2}, turn out to be increasing with an increasing $\gamma$ and with a decreasing $\alpha$.
This holds thanks to the assumption $F_0(0)<\alpha$, i.e.~$F_0^{-1}(\alpha)>0$.

\begin{table}[ht]
\caption{Illustration of Section~\ref{sec:pareto}. Values of ${\sf CVaR}_{\gamma,\alpha}$ are reported for the model (5.7) for various values of $\gamma$ and $\alpha$.}
\label{tab2}
\centering
\begin{tabular}{l|cc|ccccc}
&&& \multicolumn{5}{c}{Value of $\gamma$}\\
& $F_0^{-1}(\alpha)$ & ${\sf CVaR}_{F_0,\alpha}$ & 2&3&4&5&$+\infty$\\\hline
$\alpha=0.9$  & 2.706 & 4.39 & 1.68 & 3.04 & 3.49 & 3.71 & 4.39\\ 
$\alpha=0.95$ & 3.841 & 5.58 & 1.74 & 3.66 & 4.30 & 4.62 & 5.58\\ 
$\alpha=0.99$ & 6.635 & 8.40 & 1.77 & 5.08 & 6.19 & 6.74 & 8.40\\ 
\end{tabular}\\
\end{table}	

\subsection{Huber-type contamination with heavy tails}%%%%%%%%%%%%%%%%%%%%%%%%%%%%%%%%%%%%%%%%%%%%%%%%%%%%%%%%%%%%%% 5.2
\label{sec:huber}

Let us consider the contaminated model with a fixed $\alpha$ and with $F_0(0)<\alpha<1$ in the form
\begin{eqnarray}\label{contam2}
\widetilde{F}_{\gamma,\alpha}(x)&=&(1-\varepsilon)F_0(x)+\varepsilon F_{\gamma,\alpha}(x)\\[3mm]
&=&\left\{ \begin{array}{lll}
F_0(x)&\ldots &x\leq F_0^{-1}(\alpha)\\
%\varepsilon 
F_0(x)+\varepsilon(1-\alpha)\left(F_0^{-1}(\alpha)\Big /x\right)^{\gamma}&\ldots &x > F_0^{-1}(\alpha)\nonumber\\
\end{array}\right . \quad  
\end{eqnarray}
with fixed $\gamma>1$ and $0<\varepsilon<1$.
The probability measure induced by $\widetilde{F}_{\gamma,\alpha},$ distorted by the function $\Phi$ of~(\ref{choquet3}), has the form
\begin{eqnarray}\label{Choquet10}
&&w\left(\omega: X_{\gamma,\alpha}(\omega)>x\right)=\Phi\left(1-\widetilde{F}_{\gamma,\alpha}(x)\right)\\
&&\qquad =\left\{ \begin{array}{lll}
1 & \ldots & x\leq F_0^{-1}(\alpha)\\[2mm]
(1-\varepsilon)(1-\alpha)^{-1}\left(1-F_0(x)\right)+\varepsilon \left(F_0^{-1}(\alpha))\Big/ x\right)^{\gamma} &\ldots & x> F_0^{-1}(\alpha).\\
\end{array}\right.  \nonumber  
\end{eqnarray}
This will further lead to the conditional measure of risk ${\sf CVaR}_{\gamma, \alpha}$ corresponding to the contaminated model $\widetilde{F}_{\gamma,\alpha}$. 
Here, the notation ${\sf CVaR}_{\gamma, \alpha}$ is used to stress the dependence of the expected shortfall on $\gamma$.
The measure will be now evaluated in comparison with ${\sf CVaR}_{F_0,\alpha}$, which is the measure for the non-contaminated $F_0$. 

\begin{theorem}
\label{th:huber} 
The expected shortfall of the model (\ref{contam2}) with $F_0(0)\leq\alpha, \; \gamma>1$, and $0\leq\varepsilon<1,$ contaminated by the heavy tail with index $\gamma$, is equal to
\begin{eqnarray}\label{Choquet11}
{\sf CVaR}_{\gamma, \alpha}
&=&(1-\varepsilon)\frac{1}{1-\alpha}\int_{\alpha}^1 F_0^{-1}(t)dt+\varepsilon\frac{\gamma}{\gamma-1}~ F_0^{-1}(\alpha)\\
&=&(1-\varepsilon){\sf CVAR}_{F_0,\alpha} + \varepsilon\frac{\gamma}{\gamma-1} F_0^{-1}(\alpha),
\nonumber
\end{eqnarray}
i.e.~the effect of the heavy tail on the expected shortfall decreases with $\gamma>1$ for $F_0^{-1}(\alpha)>0.$
\end{theorem}

\begin{proof}
Denote $\tau_{\alpha}=F_0^{-1}(\alpha).$ Then, indeed,
\begin{eqnarray*} 
{\sf CVaR}_{\gamma,\alpha}&=& \tau_{\alpha}+(1-\varepsilon)(1-\alpha)^{-1}\left[-\tau_{\alpha}(1-\alpha)+\int_{\tau_{\alpha}}^{\infty} xdF_0(x)\right]
+\varepsilon \tau_{\alpha}^{\gamma}\left[\frac{x^{-\gamma+1}}{1-\gamma}\right]_{\tau_{\alpha}}^{\infty}\\
&=&\tau_{\alpha}+(1-\varepsilon)\left[-\tau_{\alpha}+(1-\alpha)^{-1}
\int_{\tau_{\alpha}}^{\infty} xdF_0(x)\right]+\varepsilon\tau_{\alpha}(\gamma-1)^{-1}\\
&=&(1-\varepsilon){\sf CVaR}_{F_0,\alpha}+\varepsilon~\tau_{\alpha}~\frac{\gamma}{\gamma-1}.
\end{eqnarray*} 
\end{proof}   

%\subsection{Numerical illustration of Section~\ref{sec:huber}.} %%%%%%%%%%%%%%%%%%%%
\noindent {\bf Numerical illustration.} 
Values of ${\sf CVaR}_{\gamma,\alpha}$ of Theorem~\ref{th:huber} will be now illustrated on a~numerical example.
Let us assume $F_0$ to be the distribution function of $\chi^2_1$ distribution.
The values of ${\sf CVaR}_{\gamma,\alpha}$ for various values of $\gamma$ and $\varepsilon$ for a fixed value $\alpha=0.96$ are shown in Table~\ref{tab3}.
For this choice of $\alpha$, we have $F_0^{-1}(\alpha)=4.218$ and ${\sf CVaR}_{F_0,\alpha}=5.98$ independently on $\varepsilon$.
The values of ${\sf CVaR}_{\gamma,\alpha}$ turn out to be increasing with an increasing $\varepsilon$ and with a~decreasing $\gamma$ (for a fixed $\varepsilon>0$).

\begin{table}[ht]
\caption{Illustration of Section~\ref{sec:huber}. Values of ${\sf CVaR}_{\gamma,\alpha}$ according to (5.8) for various values of $\gamma$ and $\varepsilon$.}
\label{tab3}
\centering
\begin{tabular}{l|ccccc}
& \multicolumn{5}{|c}{Value of $\gamma$}\\
& 1.5 & 2 & 3 & 5 & $+\infty$\\
%\begin{tabular}{l|p{18mm}p{18mm}p{18mm}p{18mm}p{18mm}}
%\begin{tabularx}{c *{5}{Y}}
%& $\gamma=1.5$ & $\gamma=2$ & $\gamma=3$ & $\gamma=5$ & $\gamma=+\infty$\\
\hline
%$\varepsilon=0$    & $-0.793$ & $-0.793$ & $-0.793$ & $-0.793$ & $-0.793$\\
%$\varepsilon=0.01$ & $-0.658$ & $-0.701$ & $-0.722$ & $-0.732$ & $-0.743$\\
%$\varepsilon=0.1$ & 0.552 & 0.130 & $-0.081$ & $-0.186$ & $-0.292$\\
%$\varepsilon=0.2$ & 1.869 & 1.053 & 0.631 & 0.420 & 0.209\\
%$\varepsilon=0.3$ & 3.241 & 1.976 & 1.343 & 1.027 & 0.710\\
$\varepsilon=0$    & 5.98&5.98&5.98&5.98&5.98\\
$\varepsilon=0.01$ & 6.05&6.00&5.98&5.97&5.96\\
$\varepsilon=0.1$  &6.65&6.23&6.01&5.91&5.80\\
$\varepsilon=0.2$  &7.31&6.47&6.05&5.84&5.63\\
$\varepsilon=0.3$  &7.98&6.72&6.08&5.77&5.45\\
\end{tabular}
\end{table}	

\section*{Conclusion} %%%%%%%%%%%%%%%%%%%%%%%%%%%%%%%%%%%%%%%%%%%%%%%%%%%%%%%%%%%%%%%%%%%%%%%%%%%%%%%%%%%%%%%%%%%%%%%

This paper is interested in estimation of the expected shortfall, i.e.~in obtaining empirical versions of one of the most popular risk measures in financial applications.
In Section~3, we propose a new nonparametric estimator of the expected shortfall for the situation with an~unknown data distribution.  
It is more complicated to estimate the expected shortfall in non-standard situations, such as in models with measurement errors or under contamination of the data distribution.

This paper proposes estimators of the expected shortfall for three realistic non-standard situations. 
For the model with additive measurement errors, the method of Choquet capacities allowed us to evaluate the upper bound for the expected shortfall. This upper bound is valid for any values of the (unknown) measurement errors.
Further, new evaluations of the expected shortfall are derived for two versions of data contamination (Section 5). The obtained estimators evaluate the influence (increase or decrease) of data contamination on the expected shortfall.
To conclude, practitioners should be aware that the expected shortfall is always obtained in a certain context and may be affected by non-standard situations such as violations of standard assumptions.

\end{document}